%% file: main.tex
\documentclass[a4paper]{article}

\input{preamble.tex}

\begin{document}

\input{TitlePage.tex}

\maketitle

\vspace{3cm}
\input{abstact.tex}
\thispagestyle{empty}

\newpage

\input{introduction.tex}

\input{basic_counting.tex}

\input{basic_summing.tex}

\input{discussion.tex}

\subparagraph*{Acknowledgments}

We thank Dror Rawitz for helpful comments. This work was partially funded by MOST grant \#3-10886 and the Technion-HPI research school.

\newpage
\input{bibliography.tex}

\newpage
\appendix
\input{basic_counting_lemma.tex}

\input{randomized_bs.tex}
\input{basic_summing_error.tex}

\input{basic-summing-memory.tex}

\input{small-eps-correctness.tex}

\input{small-eps-summing-memory-requirements.tex}







\end{document}

%% file: preamble.tex
\usepackage{fullpage}

\usepackage{microtype}

\input{macros.tex}

\usepackage{comment}
\usepackage{url}
\usepackage{dsfont}
\usepackage{amsmath,amssymb}
\usepackage{float}
\usepackage{multirow}
\usepackage{hhline}
\usepackage{calc}
\usepackage{amsthm}

\usepackage{graphicx}
\usepackage{framed}
\usepackage{xspace}
\usepackage{algorithm}
\usepackage[noend]{algpseudocode}
\usepackage{algorithm}
\usepackage{array}

\newtheorem{theorem}{Theorem}[]
\newtheorem{lemma}[theorem]{Lemma}

\newtheorem{corollary}[theorem]{Corollary}
\newtheorem{thm}{Theorem}[section] 
\newtheorem{definition}[theorem]{Definition}

\theoremstyle{plain} 
\newcommand{\thistheoremname}{}
\newtheorem{genericthm}[thm]{\thistheoremname}

\renewcommand{\qed}{\nobreak \ifvmode \relax \else
	\ifdim\lastskip<1.5em \hskip-\lastskip
	\hskip1.5em plus0em minus0.5em \fi \nobreak
	\vrule height0.75em width0.5em depth0.25em\fi}

\def\_{\,\,\,\,\,}

\newcommand{\eps}{\epsilon}

\newcommand{\set}[1]{\left\{#1\right\}}

\usepackage[noend]{algpseudocode}
\usepackage{multicol}
\usepackage{amsthm}
\usepackage{amsmath} 

%% file: macros.tex
\newcommand{\ecor}{$\epsilon$-correct}
\newcommand{\ceil}[1]{ \left\lceil{#1}\right\rceil}
\newcommand{\floor}[1]{ \left\lfloor{#1}\right\rfloor}
\newcommand{\parentheses}[1]{ \left({#1}\right)}
\newcommand{\abs}[1]{ \left|{#1}\right|}
\newcommand{\logp}[1]{\log\parentheses{#1}}
\newcommand{\cdotpa}[1]{\cdot\parentheses{#1}}
\newcommand{\inc}[1]{$#1 = #1 + 1$}
\newcommand{\range}[2][0]{#1,1,\ldots,#2}
\newcommand{\frange}[1]{\set{\range{#1}}}

\newcommand{\smallMultError}{(1+o(1))}
\newcommand{\lowerbound}{\max \set{\log W ,\frac{1}{2\epsilon+W^{-1}}}}
\newcommand{\smallEpsLowerbound}{\window\logp{\frac{1}{\weps}}}

\newcommand{\smallEpsMemoryConsumption}{W\cdot\logp{\frac{1}{\weps}}}

\newcommand{\largeEpsRestriction}{For any \largeEps{},}
\newcommand{\largeEps}{$\eps^{-1} \le 2W\left(1-\frac{1}{\logw}\right)$}
\newcommand{\smallEpsRestriction}{For any \smallEps{},}
\newcommand{\smallEps}{$\eps^{-1}>2W\left(1-\frac{1}{\logw}\right)=2\window(1-o(1))$}
\newcommand{\bc}{{\sc Basic-Counting}}
\newcommand{\bs}{{\sc Basic-Summing}}

\newcommand{\query}[1][] { {\sc Query}$(#1)$}
\newcommand{\add}  [1][] { {\sc Add}$(#1)$}

\newcommand{\window}{W}
\newcommand{\logw}{\log \window}
\newcommand{\flogw}{\floor{\log \window}}
\newcommand{\weps}{\window\epsilon}
\newcommand{\logweps}{\logp{\weps}}
\newcommand{\bitarray}{b}
\newcommand{\currentBlockIndex}{i}
\newcommand{\currentBlock}{\bitarray_{\currentBlockIndex}}
\newcommand{\remainder}{y}
\newcommand{\numBlocks}{k}
\newcommand{\sumOfBits}{B}
\newcommand{\blockSize}{\frac{\window}{\numBlocks}}
\newcommand{\iblockSize}{\frac{\numBlocks}{\window}}
\newcommand{\threshold}{\blockSize}
\newcommand{\halfBlock}{\frac{\window}{2\numBlocks}}
\newcommand{\blockOffset}{m}
\newcommand{\inputVariable}{x}

\newcommand{\bsTableColumnWidth}{1.7cm}

\newcommand{\bcNarrowTableColumnWidth}{1.5cm}

\newcommand{\bsrange}{ R }

\newcommand{\bssum}{ S^W }
\newcommand{\bsFracInput}{ \inputVariable' }
\newcommand{\bserror}{ \bsrange\window\epsilon }

\newcommand{\bsReminderFractionBits}{ \upsilon}
\newcommand{\bsAnalysisTarget}{ \bssum}
\newcommand{\bsAnalysisEstimator}{ \widehat{\bsAnalysisTarget}}
\newcommand{\bsAnalysisError}{ \bsAnalysisEstimator - \bsAnalysisTarget}
\newcommand{\bsRoundingError}{ \xi}

%% file: TitlePage.tex
\title{Efficient Summing over Sliding Windows}

\author{Ran Ben Basat \ \ \ \ Gil Einziger \ \ \ \ Roy Friedman \ \ \ \ Yaron Kassner\\
    Department of Computer Science\\
    Technion, Haifa 32000,~Israel\\
    \texttt{\{sran, gilga ,roy, kassnery\}}@cs.technion.ac.il}

\date{}

\setcounter{page}{0}
\thispagestyle{empty}

%% file: abstact.tex
\begin{abstract}
This paper considers the problem of maintaining statistic aggregates over the last $\window$ elements of a data stream.
First, the problem of counting the number of $1$'s in the last $\window$ bits of a binary stream is considered.
A lower bound of $\Omega(\frac{1}{\eps}+\log\window)$ memory bits for $\weps$-\emph{additive approximations} is derived.
This is followed by an algorithm whose memory consumption is $O(\frac{1}{\eps}+\log\window)$ bits, indicating that the algorithm is optimal and that the bound is tight.
Next, the more general problem of maintaining a \emph{sum} of the last $W$ integers, each in the range of
$\frange\bsrange$,
is addressed.
The paper shows that approximating the sum within an \emph{additive error} of $\bserror$ can also be done using $\Theta(\frac{1}{\eps}+\log\window)$ bits for $\eps=\Omega\parentheses{\frac{1}{\window}}$. For $\eps=o\parentheses{\frac{1}{\window}}$, we present a \emph{succinct} algorithm which uses $\mathcal{B}\cdot\smallMultError$ bits, where $\mathcal{B}=\Theta\parentheses{\window \logp{\frac{1}{\weps}}}$ is the derived lower bound. 
We show that all lower bounds generalize to randomized algorithms as well.
All algorithms process new elements and answer queries in $O(1)$ worst-case time.
\end{abstract}

%% file: introduction.tex
\section{Introduction}
\label{sec:intro}

\input{intro-background.tex}

\input{our-contribution.tex}

\input{related-work.tex}

%% file: intro-background.tex
\subsubsection*{Background}
The ability to process and maintain statistics about large streams of data is useful in many domains, such as security, networking, sensor networks, economics, business intelligence, etc.
Since the data may change considerably over time, there is often a need to keep the statistics only with respect to some window of the last $\window$ elements at any given point.
A~naive solution to this problem is to keep the $W$ most recent elements, add an element to the statistic when it arrives, and subtract it when it leaves the window.
Yet, when the window of interest is large, which is often the case when data arrive at high rate, the required memory overhead may become a performance bottleneck.

Though it may be tempting to think that RAM memory is cheap, a closer look indicates that there are still performance benefits in maintaining small data structures.
For example, hardware devices such as network switches prefer to store important data in the faster and scarcely available SRAM than in DRAM.
This is in order to keep up with the ever increasing line-speed of modern networks.
Similarly, on a CPU, caches provide much faster performance than DRAM memory.
Thus, small data structures that fit inside a single cache line and can possibly be pinned there are likely to result in much faster performance than a solution that spans multiple lines that are less likely to be constantly maintained in the cache. 

A well known method to conserve space is to approximate the statistics.
\bc{} is one of the most basic textbook examples of such approximated stream processing problems~\cite{DatarGIM02}.
In this problem, one is required to keep track of the number of $1$'s in a stream of binary bits.
A \emph{$(1+\eps)$-multiplicative approximation} algorithm for this problem using $O\left(\frac{1}{\eps}\log^2\window\eps\right)$ bits was shown in~\cite{DatarGIM02}.
This solution works with amortized $O(1)$ time, but its worst case time complexity is $O(\logw)$.

A more practical related problem is \bs{}, in which the goal is to maintain the sum of the last $\window$ elements.
When all elements are non-negative integers in the range $[\bsrange+1]=\frange\bsrange$, the work in~\cite{DatarGIM02} naturally extends to provide a $(1+\eps)$-multiplicative approximation of this problem
using
$O\parentheses{\frac{1}{\eps}\cdot \parentheses{\log^2\window +\log\bsrange\cdotpa{\logw+\log\log\bsrange}}}$
bits.
The amortized time complexity becomes $O(\frac{\log\bsrange}{\logw})$ and the worst case is $O(\logw+\log\bsrange)$.

%% file: our-contribution.tex
\subsubsection*{Our Contributions}
In this paper, we explore the benefits of changing the approximation guarantee from \emph{multiplicative} to \emph{additive}.
With a multiplicative approximation, the result returned can be different from the correct one by at most a multiplicative factor, e.g., $5\%$.
On the other hand, in an additive approximation, the absolute error is bounded, e.g., a deviation of up to $\pm 5$.
When the expected number of ones in a stream is small, multiplicative approximation is more appealing, since its absolute error is small.
However, in this case, an accurate (sparse) representation can be even more space efficient than the multiplicative approximation.
On the other hand, when many ones are expected, additive approximation gives similar outcomes to multiplicative approximation.
Furthermore, the potential space saving becomes significant in this case, motivating our exploration.

Our initial contribution is a formally proved memory lower bound of $\Omega(\frac{1}{\eps}+\logw)$ for $\weps$-additive approximations for the \bc{} problem.

Our second contribution is a space optimal algorithm providing a \emph{$\weps$-additive approximation} for the \bc{} problem.
It consumes $O(\frac{1}{\eps}+\logw)$ memory bits with a worst case time complexity of $O(1)$, matching the lower bound.

Next, we explore the more general \bs{} problem.
Here, the results are split based on the value of $\eps$.
Specifically, our third contribution is an (asymptotically) space optimal algorithm providing an $\bserror$-additive approximation for the \bs{} problem when \largeEps.\footnote{In this paper, the logarithms are of base 2 and the o(1) notation is for $W\to\infty$.}
It uses $O(\frac{1}{\eps}+\logw)$ memory bits and has $O(1)$ worst case time complexity.
For other values of $\epsilon$, we show a lower bound of $\Omega(\window\logp{\frac{1}{\weps}})$ and a corresponding algorithm requiring $O\parentheses{\window\logp{\frac{1}{2W\eps}+1}}$ memory bits
with $O(1)$ worst case time complexity. Furthermore, we show that this algorithm is succinct for $\epsilon=o(W^{-1})$, i.e. its space requirement is only \smallMultError{} times the lower bound.


To get a feel for the applicability of these results, consider for example an algorithmic trader that makes transactions based on a moving average of the gold price. He samples the spot price once every millisecond, and wishes to approximate the average price for the last hour, i.e., $\window=3.6\cdot 10^6$ samples. The current price is around $\$1200$, and with a standard deviation of $\$10$, he safely assumes the price is bounded by $\bsrange\triangleq 1500$. The trader is willing to withstand an error of $0.1\%$, which is approximately $\$1.2$. Our algorithm provides a $\window\bsrange\eps_A$ (the ‘$A$’ stands for \emph{Additive}) additive-approximation using $\parentheses{\frac{1}{2\eps_A}+2\logw}(1+o(1))$ memory bits, while the  algorithm by Datar et al. \cite{DatarGIM02} computes a multiplicative $(1+\eps_M)$ (the ‘$M$’ stands for \emph{Multiplicative}) approximation using $\ceil{\frac{1}{2\eps_M}+1}\ceil{\logp {2\window\bsrange\eps_M+1}+1}$ buckets of size $\ceil{\logw+\logp{\logw+\log\bsrange}}$ bits each.
Using our algorithm, the trader sets $\eps_A=\bsrange^{-1}=\frac{1}{1500}$, which guarantees that as long as the price of gold stays above $\$1000$, the error remains lower than required.
The multiplicative approximation algorithm requires setting $\eps_M=0.1\%$, and uses $501\cdot\ceil{\logp {1080001}+1}=12525$ buckets of size $27$ bits each and about $41\mathit{KB}$ overall.
In comparison, our algorithm with the parameters above requires only about $100$ bytes.

Another useful application for our algorithm is counting within a \emph{fixed} additive error.
The straight-forward algorithm for solving \bc{} uses a $\window$-bits array which stores the entire window, replacing the oldest recorded bit with a new one whenever such arrives.
Assume a $\pm 5$ error is allowed. Using the multiplicative-approximation algorithms, one has to set $\eps_M=\frac{5}{\window}$, which requires more than $W$ bits, worse than exact counting.
In contrast, setting $\eps_A=\frac{5}{\window}$ for our algorithm reduces the memory consumption of the exact solution by nearly $90\%$.

In summary, we show that additive approximations offer significant space reduction opportunities.
They can be obtained with a constant worst case time complexity, which is important in real-time and time sensitive applications.

%% file: related-work.tex
\section{Related Work}

In \cite{DatarGIM02}, Datar et al. first presented the problem of counting the number of $1$'s in a sliding window of size $W$ over a binary stream, and its generalization to summing a window over a stream of integers in the range $\set{0,1,\ldots,R}$. They have introduced a data structure called \emph{exponential histogram} $(\mathit{EH})$. $\mathit{EH}$ is a time-stamp based structure that partitions the stream into \emph{buckets}, saving the time elapsed since the last $1$ in the bucket was seen. Using $\mathit{EH}$, they have derived a space-optimal algorithm for approximating \bs{} within a multiplicative-factor of $(1+\eps)$, which uses $O\parentheses{\frac{1}{\eps}\log^2\window +\log\bsrange\cdotpa{\logw+\log\log\bsrange}}$ memory bits.
The structure allows estimating a class of aggregate functions such as counting, summing and computing the $\ell_1$ and $\ell_2$ norms of a sliding window in a stream containing integers. The exponential histogram technique was later expanded \cite{BabcockDMO03} to support computation of additional functions such as $k$-median and variance. Gibbons and Tirthapura \cite{GibbonsT02} presented a different structure called \emph{waves}, which improved the worst-case runtime of processing a new element to a constant, keeping space requirement comparable when $R=poly(W)$. Braverman and Ostrowsky~\cite{SmoothHistograms} defined \emph{smooth histogram}, a generalization of the exponential histogram, which allowed estimation of a wider class of aggregate functions and improved previous results for several functions such as $l_p$ norms and frequency moments. Lee and Ting~\cite{LeeT06} presented an improved algorithm, requiring less space if a $(1+\eps)$ approximation is guaranteed only when the ones consist of a significant fraction of the window. They also presented the $\lambda$ counter~\cite{lambdaCounter} that counts bits over a sliding window as part of a frequent items algorithm.
Our design is more space efficient as it requires
$O(\frac{1}{\varepsilon} + log(n))$ bits instead of  $O(\frac{1}{\varepsilon}\cdot log(n))$ bits.

In~\cite{CohenS06}, Cohen and Strauss considered a generalization of the bit-counting problem on a sliding window for computing a weighted sum for some decay function, such that the more recent bits have higher weights. Cormode and Yi~\cite{DistributedAggregates} solved bit counting in a distributed setting with optimal communication between nodes. Table 1 and
Table 2 summarize previous works on the \bc{} and \bs{} problems and compare them to our own algorithms.

Extensive studies were conducted on many other streaming problems over
sliding windows such as Top-K~\cite{TopKSlidingWindow,TopK2}, Top-K tuples~\cite{TopKPairsSlidingWindows},
Quantiles~\cite{Quantitiles}, heavy hitters~\cite{HHINFOCOM,L2HeavyHitters,HeavyHitters}, distinct items ~\cite{DistinctSlidingWindow}, duplicates~\cite{duplicate}, Longest Increasing Subsequences~\cite{SmoothHistograms,LIS}, Bloom filters~\cite{SlidingBloomInfocom,SlidingBloomFilter}, graph problems~\cite{WeightedMatching,Graphs} and more.

%% file: basic_counting.tex
\section{Basic-Counting Problem}\label{sec:basic-counting}

\input{basic-counting-comparison-table.tex}

\begin{definition}[Approximation]
Given a value $V$ and a constant $\epsilon$, we say that $\widehat{V}$ is an \emph{$\epsilon$-multiplicative approximation} of $V$ if $|V - \widehat{V}| < \epsilon V$.
We say that $\widehat{V}$ is an \emph{$\epsilon$-additive approximation} of $V$ if $|V - \widehat{V}| < \epsilon$.
\end{definition}

\begin{definition}[\bc]
Given a stream of bits and a parameter $W$, maintain the number of $1$'s in the last $W$ bits of the stream.
Denote this number by $C^W$.
\end{definition}

\input{basic-counting-lower-bound.tex}

\input{basic-counting-upper-bound-and-optimality.tex}

%% file: basic-counting-comparison-table.tex
\begin{table*}[t]
\centering
\footnotesize
\centering{
\tabcolsep=0.11cm
\begin{tabular}{|p{\widthof{Tirthapura [11]}}|p{\widthof{ Approximation Guarantee }}|p{\widthof{Requirement......}}|p{\bcNarrowTableColumnWidth}|p{\bcNarrowTableColumnWidth}|p{\widthof{Additive.}}|}
	\hline
	\bc{}& Approximation Guarantee & Memory \mbox{Requirement} & Amortized Addition Time & Worst-Case Addition Time & Maximal Additive Error \\ \hline
	Datar et al.~\cite{DatarGIM02}& $\begin{aligned}(1+\eps)\end{aligned}$-Multiplicative & $O\bigg(\frac{1}{\eps}\log^2\weps\bigg)$ & $O(1)$ & $O(\logw)$ &$\weps$\\ \hline
	Gibbons~and Tirthapura~\cite{GibbonsT02}& $\begin{aligned}(1+\eps)\end{aligned}$-Multiplicative & $O\bigg(\frac{1}{\eps}\log^2\weps\bigg)$ & $O(1)$ & $O(1)$ &$\weps$\\ \hline
	Lee and Ting~\cite{LeeT06}& $\begin{aligned}(1+\eps)\end{aligned}$-Multiplicative, whenever there are at least
	$\theta\window$~1-bits &
	$\begin{aligned}
		O\bigg(&\frac{1}{\eps}\log^2{\frac{1}{\theta}}\\
		+& \log\window\theta\epsilon\bigg)
	\end{aligned}$ & $O(1)$ & $O(1)$ &$\weps$\\ \hline	
	This Paper& $\weps$-Additive &
	$O\bigg(\frac{1}{\eps}+\log\weps\bigg)$ &
	$O(1)$ & $O(1)$&$\weps$\\ \hline	
\end{tabular}\smallskip}
\normalsize
\caption{Comparison of \bc{} Algorithms.}
\label{tab:bc}
\end{table*}
\begin{table*}[tp]
	\centering
	\footnotesize
	\tabcolsep=0.11cm
	\begin{tabular}{|p{\widthof{Tirthapura [12]}}|p{\widthof{Approximation}}|p{\widthof{Memory Requirement............} }|p{\bsTableColumnWidth}|p{\widthof{Addition~Time}}|p{\widthof{ Additive }}|}
		\hline
		\bs{}& Approximation Guarantee & Memory Requirement & Amortized Addition Time   			& Worst-Case Addition Time	& Maximal Additive Error			 \\ \hline
		Datar et al.~\cite{DatarGIM02}& $\begin{aligned}(1+\eps)\end{aligned}$-Multiplicative &
		$\begin{aligned}
		&O\left(\frac{1}{\eps}\big(\log^2\window\right .\\
		&+\log\bsrange\logw \\
		&+\log\bsrange\log\log\bsrange\big)\bigg)
		\end{aligned}$
		& $O\parentheses{\frac{\log\bsrange}{\logw}}$ & $\begin{aligned}O(&\logw\\+&\log\bsrange)\end{aligned}$ & $\bserror$
		\\ \hline
		Gibbons~and Tirthapura~\cite{GibbonsT02}& $\begin{aligned}(1+\eps)\end{aligned}$-Multiplicative &
		$\begin{aligned}
		&O\left(\frac{1}{\eps}\parentheses{\logw+\log\bsrange}^2\right)
		\end{aligned}$
		& $O(1)$ & $O(1)$ & $\bserror$
		\\ \hline
		\multirow{5}{*}{This Paper}& $\bserror$-Additive for $\epsilon\ge\frac{1}{2\window}$ & \raisebox{-0.4cm}{$O\left(\frac{1}{\eps}+\logw\right)$} & \multirow{5}{*}{$O(1)$} & \multirow{5}{*}{$O(1)$}& \multirow{5}{*}{$\bserror$}\\
		\hhline{~--}& $\bserror$-Additive for $\epsilon\le\frac{1}{2\window}$ & \raisebox{-0.4cm}{$O\parentheses{\smallEpsMemoryConsumption}$} & & & \\ \hline		
	\end{tabular}\smallskip
	\normalsize
	\caption{Comparison of \bs{} Algorithms.}
	\label{tab:bs}
\end{table*} 

%% file: basic-counting-lower-bound.tex
\subsection{Lower Bound}

We now show lower bounds for the memory requirement for approximating \textsc{Basic-Counting}.
\begin {lemma}\label{lem:approx-counting-1}
For any $\eps$ and $W$, any deterministic algorithm that provides a $W\epsilon$-additive approximation for \textsc{Basic-Counting} requires at least
$\left\lfloor\frac{W}{\left\lfloor2W\epsilon+1\right\rfloor}\right\rfloor \ge \floor{\frac{1}{2\epsilon+W^{-1}}}$ bits.
\end{lemma}

\begin{proof}
Denote $z\triangleq\left\lfloor\frac{W}{\left\lfloor2W\epsilon+1\right\rfloor}\right\rfloor$.
We prove the lemma by showing $2^z$ arrangements that must lead to different configurations.
Consider the language of all concatenations of $z$ blocks of size $\left\lfloor2W\epsilon+1\right\rfloor$, such that each block consists of only ones or only zeros:
$$
L_{W,\epsilon}=\{w_0w_1\cdots w_{z-1} \mid 
\forall j\in [z]: w_j=0^{\left\lfloor2W\epsilon+1\right\rfloor}
\vee w_j=1^{\left\lfloor2W\epsilon+1\right\rfloor}\} 
$$
Assume, by way of contradiction, that two different words
$$s^1=w_0^1w_1^1\cdots w_{z-1}^1, s^2=w_0^2w_1^2\cdots w_{z-1}^2\in L_{W,\epsilon}$$
lead the algorithm to the same configuration.
Denote the index of the last block that differs between $s^1$ and $s^2$
by $t \triangleq max\{\tau\mid w_\tau^1\neq w_\tau^2\}.$
Next, consider the sequences
$s^1\cdot 0^{(t-1)\left\lfloor 2W\epsilon+1\right\rfloor}$ and
$s^2\cdot 0^{(t-1)\left\lfloor 2W\epsilon+1\right\rfloor}.$
The algorithm must reach the same configuration after processing these sequences, even though the number of ones differs by $\left\lfloor2W\epsilon+1\right\rfloor>2W\epsilon$.
Therefore, the algorithm's error must be greater than $W\epsilon$ at least for one of the sequences, in contradiction to the assumption.
We have shown $2^z$ words that lead to different configurations and
therefore any deterministic algorithm that provides $\epsilon-additive$ approximation to \textsc{Basic-Counting} must have at least $z$ bits of state.
%
\end{proof}

An immediate corollary of Lemma~\ref{lem:approx-counting-1} is that any exact algorithm for \textsc{Basic-Counting} requires at least $W$ bits, i.e., the naive solution is optimal.
We next establish a second lower bound, which is useful for proving that our algorithm, presented below, is space optimal up to a constant factor.

\begin {lemma}\label{lem:approx-counting-2}
Fix some $\epsilon\leq \frac{1}{4}$.
Any deterministic algorithm that provides a $W\epsilon$-additive approximation for the \textsc{Basic-Counting} problem requires at least $\flogw$ bits.
\end{lemma}

\begin{proof}
Assume that some algorithm $A$ gives a $W\epsilon$-additive approximation using $m$ memory bits.
Consider $A$'s run on the sequence $s=0^W\cdot 1^{2^m}$.
Since $A$ is using $m$ bits, it reaches some memory configuration $c$ at least twice after processing the zeros in the sequence.
Assume that $A$ first reached $c$ after seeing $0^W\cdot 1^y$ (where $y<2^m$).
This means that $A$ must output some number $a_c\leq y+W\epsilon$ if queried.
Now assume $A$ returns to configuration $c$ after reading $z$ additional ones.
This means $A$ will return to $c$ after every additional sequence of $z$ ones.
Therefore, for every integer $q$, after processing the sequence $0^W\cdot 1^{y+qz}$, $A$ will reach configuration $c$.
We can then pick a large $q$ (such that $y+qz\geq W$), which means that the query answer for configuration $c$, $a_c$, has to be at least $W(1-\epsilon)$, as the window is now all-ones.
We get
$W(1-\epsilon)\leq a_c\leq y+W\epsilon$
and thus $2^m > y \geq W(1-2\epsilon)$.
Putting everything together, we conclude that
$m>\log \left(W(1-2\epsilon)\right)=\log W + \log\left(1-2\epsilon\right)\geq \log W - 1$, for $\epsilon\le\frac{1}{4}$.
Finally, since $m$ is an integer, this implies $m\ge\flogw$.
\end{proof}

\begin{theorem}\label{thm:BasicCountingLB}
Let $\eps\le\frac{1}{4}$. Any deterministic algorithm that provides a $W\epsilon$-additive approximation for the \textsc{Basic-Counting} problem requires at least
$\floor{\lowerbound}$~bits.
\end{theorem}
\begin{proof}
Immediate from lemmas \ref{lem:approx-counting-1} and \ref{lem:approx-counting-2}.
\end{proof}

Finally, we extend our lower bounds to randomized algorithms.

\begin{theorem}\label{thm:RandomizedBasicCountingLB}
	Let $\eps\le\frac{1}{4}$. Any randomized Las Vegas algorithm that provides a $W\epsilon$-additive approximation for the \textsc{Basic-Counting} problem requires at least
	$\floor{\lowerbound}$ bits. Further, for any fixed $\delta\in(0,1/2)$, any Monte Carlo algorithm that  with probability at least $1-\delta$ approximates \bc{} within $W\epsilon$ error at any time instant, requires $\Omega(\frac{1}{\eps}+\logw)$ bits.
\end{theorem}

\input{randomized_bc.tex}

%% file: randomized_bc.tex

\begin{proof}
We say that algorithm $A$ is \ecor{} on a input instance $S$ if it is able to approximate the number of $1$'s in the last $W$ bits, at every time instant while reading $S$, to within an additive error of $\weps$.

We remind the reader that in our case, a Las Vegas (\textit{LV}) algorithm for the \bc{} approximation problem is a randomized algorithm which is \emph{always} \ecor{}. In contrast, a Monte Carlo (\textit{MC}) algorithm is a randomized procedure that is allowed to provide approximation with error larger than $\weps$, with probability at most $\delta$.

The Yao Minimax principle~\cite{minimax} implies that the amount of memory required for a deterministic algorithm to approximate a random input chosen according to a distribution $\mathbf p$ is a lower bound on the expected space consumption of a Las Vegas algorithm for the worst input.
To prove a $\floor{\frac{1}{2\epsilon+W^{-1}}}$ lower bound, we consider padding the language $L_{W,\eps}$ which is defined in Lemma~\ref{lem:approx-counting-1}.
Specifically, we define $\mathbf p$ as the uniform distribution over all inputs in the language
$$L_{LV} = L_{W,\epsilon}\cdot \set{0^W}.$$
That is, the input consist of all bit sequences in $L_{W,\eps}$, followed by a sequence of $W$ zeros.		
The trailing $0$'s are used to force the algorithm to reach distinct configurations after reading the first $W$ input bits.
As implied by the lemma, any deterministic algorithm which is always correct for a random instance requires at least $\floor{\frac{1}{2\epsilon+W^{-1}}}$ bits, as it has to arrive to a distinct state for each input $s\in L_{W,\epsilon}$.  The argument for a lower bound of $\flogw$ bits is similar.

Next, we use the Minimax principle analogue for Monte Carlo algorithms~\cite{minimax}, which states that for any input distribution $\mathbf p$ and $\delta\in [0,1/2]$, any randomized algorithm that is always (for any input) \ecor{} with probability at least $1-\delta$ uses in expectation at least half as much memory as the optimal deterministic algorithm that errs (i.e., is not \ecor{}) with probability at most $2\delta$ on a random instance drawn according to $\mathbf{p}$.
Once again, we consider $\mathbf p$ to be the uniform distribution over
$$L_{MC} = L_{W,\epsilon}\cdot \set{0^W}.$$
Since the distribution is uniform, any deterministic algorithm, which is \ecor{} with probability at least $1-2\delta$ on a random instance drawn according to $\mathbf p$, is actually \ecor{} on $1-2\delta$ fraction of the inputs. Similar to the LV case, the argument in Lemma~\ref{lem:approx-counting-1} implies that the algorithm must reach a distinct configuration after reading the first $W$ bits of each of the $(1-2\delta)\cdot|L_{MC}|$ inputs it is \ecor{} on. Consequently, the algorithm must use at least $\logp{(1-2\delta)\cdot|L_{MC}|}$ bits of memory. Applying the Minimax principle, the derived lower bound $B_{MC}$ for any MC algorithm is:
$$B_{MC}\ge \frac{1}{2}\logp{(1-2\delta)\cdot|L_{MC}|}\ge\frac{1}{2}\floor{\frac{1}{2\epsilon+W^{-1}}}+\frac{1}{2}\logp{1-2\delta}=\Omega\parentheses{\frac{1}{\eps}}$$
Once again, the case for a $\Omega(\logw)$ lower bound is based on Lemma~\ref{lem:approx-counting-2} and follows from similar arguments.
\end{proof}

%% file: basic-counting-upper-bound-and-optimality.tex
\subsection{Upper Bound}

We now present an algorithm for \textsc{Basic-Counting} that provides a $W\epsilon$-additive approximation $\widehat{C^W}$ for $C^W$ over a binary stream with near-optimal memory.
Denote $k\triangleq \frac{1}{2\epsilon}$.
For simplicity, we assume that $\frac{W}{k}$ and $k$ are integers.
Intuitively, our algorithm partitions the stream into $k$ \emph{blocks} of size $W\over k$, representing each using a single bit. A set bit corresponds to a count of $\frac{W}{k}$ in the input stream, while a clear bit corresponds to a count of 0. We then use an ``optimistic'' approach to reduce the error -- the number of ones in the input stream not counted using the bit array is \emph{propagated} to the next block; this means a block might be represented with $1$, even if it contains only a single set bit. Surprisingly, we show that this approach allows us to keep the error bounded and that the errors do not accumulate.
We keep a counter $y$ for the number of $1$s.
At the end of a block, if $y$ is larger than $W\over k$, we mark the current block and subtract $W\over k$ from $y$, propagating the remainder to the next block.
Our algorithm answers queries by multiplying the number of marked blocks in the current window by $W\over k$, making corrections to reduce the error.
%
We maintain the following variables:
\begin{description}
\item[$y$] - a counter for the number of ones.
\item[$b$] - a bit-array of size $k$.
\item[$i$] - the index of the ``oldest'' block in $b$.
\item[$B$] - the sum of all bits in $b$.
\item[$m$] - a counter for the current offset within the block.
\end{description}
\normalsize
Every arriving bit is handled as follows: We increment $m$, and if the bit is set we also increment $y$.
At the end of a block, we check if $y$ exceeds $W \over k$. If so, we subtract $W \over k$ from $y$ and set the bit $b_i$.
This way, the reduction in $y$ is compensated for by the newly set bit in $b$.
The previous value of $b_i$, holding information about $1$s that just left the window, is forgotten.

To answer a query the algorithm returns the number of set bits in $b$ multiplied by the block size $W \over k$.
We then add the value of $y$, which represents the number of ones not yet recorded in $b$, and subtract $m\cdot b_i$, as $m$ bits of the oldest recorded block have already left the window.
Finally, we remove any bias from the estimation by subtracting $W \over 2k$ (half a block).

In order to answer queries without iterating over $b$, we maintain another variable $B$, which keeps track of the number of ones in $b$.
The entire pseudo-code is given in Algorithm~\ref{alg:BasicCounting}.

\input{basic-counting-algorithm-pseudo-code.tex}


\begin{theorem} \label{thm:approximation-ratio-bc}
Algorithm~\ref{alg:BasicCounting} provides a $W\epsilon$-additive approximation of \bc{}.
\end{theorem}
\begin{proof}
First, let us introduce some notations used in the proof.
Assume that the index of the last bit is $W+m$, where $x_W$ is the last bit of a block and $m<{W\over k}$.
$b_i$ is considered after $W+m$ bits have been processed.
We denote $y_j$ the value of $y$ \emph{after} adding bit $j$.

The setting for the proof is given in Figure \ref{fig:basic-counting}.
We aim to approximate
\begin{align}
C^W\triangleq\sum_{j=m+1}^{W+m} x_j \label{eq:to-approximate}.
\end{align}
\begin{figure}[t!]
\centering
\includegraphics[width=0.8\linewidth,natwidth=1840,natheight=609]{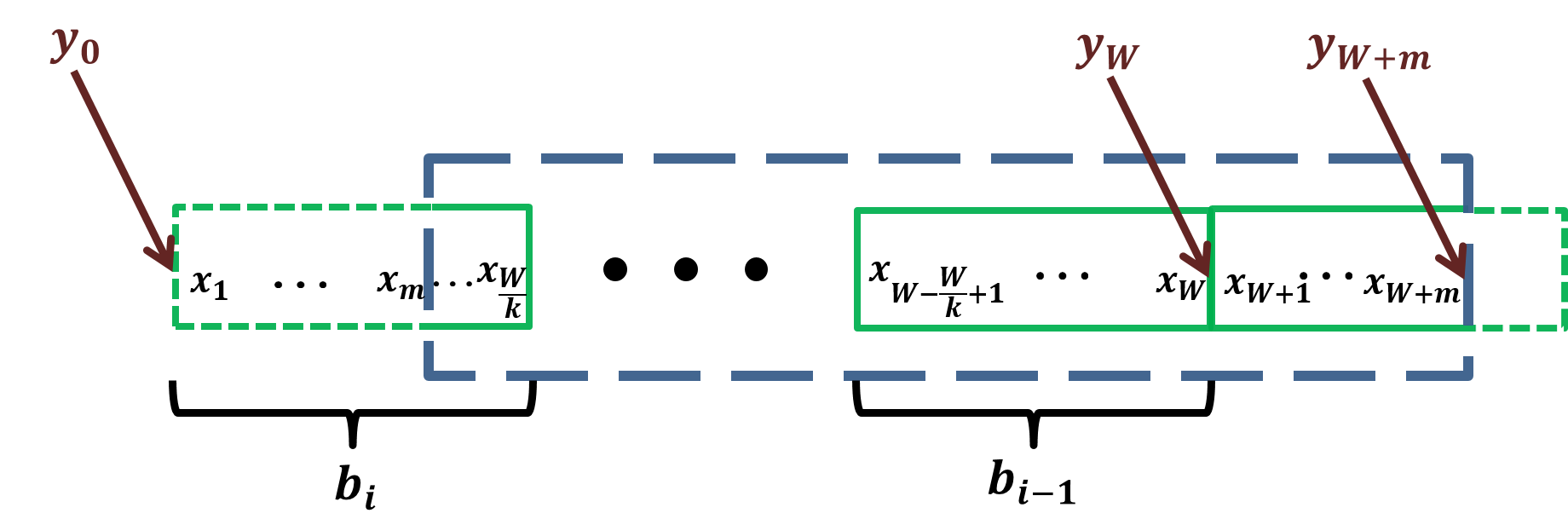}
\caption{The setting for the proof of Theorem~\ref{thm:approximation-ratio-bc}. $b$ is cyclic - $b_i$ represents the oldest block and $b_{i-1}$ the newest completed block.}
\label{fig:basic-counting}
\end{figure}
Our algorithm uses the following approximation:
\begin{align}
\widehat{C^W}
= \frac{W}{k} \cdot B + y_{W+m} - \frac{W}{2k}-m\cdot b_i 
= \frac{W}{k} \cdot B + y_W + \sum_{j=W+1}^{W+m} x_j - \frac{W}{2k}-m\cdot b_i
\label{eq:bc-approximation1}.
\end{align}
At times $1,2,\ldots ,W$, $y$ is incremented once for every set bit in the input stream.
At the end of block $j$, if $y$ is reduced by $\frac{W}{k}$, then $b_j$ is set and will not be cleared before time $W+m$.
Therefore, $\frac{W}{k} \cdot B + y_W = y_0 + \sum_{j=1}^{W} x_j.$
Substituting $\frac{W}{k} \cdot B + y_W$ in \eqref{eq:bc-approximation1}, we get
\begin{align*}
\widehat{C^W}
= y_0 + \sum_{j=1}^{W} x_j  + \sum_{j=W+1}^{W+m} x_j - \frac{W}{2k}-m\cdot b_i
= y_0 + \sum_{j=1}^{m} x_j  + \sum_{j=m+1}^{W+m} x_j - \frac{W}{2k}-m\cdot b_i.
\end{align*}

Plugging
the definition of $C^W$, we get
$\widehat{C^W} = y_0 + \sum_{j=1}^{m} x_j  + C^W - \frac{W}{2k}-m\cdot b_i$.
Therefore, the error is
$$\widehat{C^W} - C^W = y_0 + \sum_{j=1}^{m} x_j -m\cdot b_i - \frac{W}{2k}.$$
\noindent
We consider two cases:
\begin{description}
\item [$b_i = 1:$] This means that $y$ had crossed the threshold by time ${W \over k}$, i.e.
$y_0 + \sum_{j=1}^\frac{W}{k}x_j \geq \frac{W}{k}$
and equivalently $y_0 + \sum_{j=1}^{m}x_j \geq \frac{W}{k} - \sum_{j=m+1}^\frac{W}{k}x_j$.
  Thus, on one side
	\begin{align*}
	\widehat{C^W} - C^W
    = y_0 + \sum_{j=1}^{m}x_j -m - \frac{W}{2k}
    \geq \frac{W}{k} - \sum_{j=m+1}^\frac{W}{k}x_j-m - \frac{W}{2k} \\
    \geq \frac{W}{k} - \left(\sum_{j=m+1}^\frac{W}{k}1\right)-m - \frac{W}{2k}
    \geq-\frac{W}{2k}
    =-W\epsilon.
    \end{align*}
To bound the error from above we use the fact that the value of $y$ at the end of a block never exceeds $W \over k$.
This can be shown by induction, as $y$ is incremented at most $W\over k$ times during one block, and then reduced by $W \over k$ if it exceeds the block size.
Therefore,
$
\widehat{C^W} - C^W
= y_0 + \sum_{j=1}^{m} x_j -m - \frac{W}{2k} 
\leq y_0 - \frac{W}{2k}\leq\frac{W}{k} - \frac{W}{2k}= W\epsilon.\\
$
\item [$b_i = 0$:] Similarly, this means that $y$ was lower than the threshold at the end of block $i$, hence
    $y_0 + \sum_{j=1}^\frac{W}{k}x_j \leq \frac{W}{k} - 1$ or equivalently,
    $y_0 + \sum_{j=1}^{m}x_j \leq \frac{W}{k} - \sum_{j=m+1}^\frac{W}{k}x_j - 1.$
    Thus, our error is bounded from below by
    $
    \widehat{C^W} - C^W = y_0 + \sum_{j=1}^{m} x_j - \frac{W}{2k} 
    \geq y_0 - \frac{W}{2k} \geq
    - \frac{W}{2k} =
    -W\epsilon 
    $
    and from above by
    $$
    \widehat{C^W} - C^W = y_0 + \sum_{j=1}^{m} x_j - \frac{W}{2k} 
    \leq \frac{W}{k} - \sum_{j=m+1}^\frac{W}{k} x_j- \frac{W}{2k} - 1
    \leq \frac{W}{2k} - 1  \\= W\epsilon - 1.
$$
\end{description}

We have established that in all cases the absolute error is at most $W\epsilon$ as required.
\end{proof}

%

We next prove that the memory requirement of Algorithm~\ref{alg:BasicCounting} is nearly optimal.
\begin{theorem}
\label{thm:counting-memory}
Algorithm \ref{alg:BasicCounting} requires
$\frac{1}{2\epsilon} + 2\log W + O(1)$ bits of memory.
\end{theorem}

\begin{proof}
We represent $y$ using $\left\lceil 2+\log(W\epsilon)\right\rceil$ bits, $m$ using
$\left\lceil 1+log(W\epsilon)\right\rceil$ bits and $b$ using $k$ bits.
Additionally, $i$ requires $\lceil\log k\rceil$ bits, and $B$ another $\lceil\log (k+1)\rceil$ bits.
Overall, the number of bits required is
$
k + \left\lceil 2+\log(W\epsilon)\right\rceil + \left\lceil 1+\log(W\epsilon)\right\rceil + \lceil\log k\rceil + \lceil\log (k+1)\rceil\\
\le k+2\log(W\epsilon)-2\log(2\eps)+8 = \frac{1}{2\epsilon}+2\log(W)+6
= \frac{1}{2\epsilon} + 2\log W + O(1).
$
\end{proof}

Theorem \ref{thm:BasicCountingLB} shows that our algorithm uses at most twice as much memory as required by the lower bound (up to a constant number of bits) for every \emph{constant} $\epsilon\le\frac{1}{4}$.
When $\epsilon$ is not constant, our memory requirement is at most 3 times the lower bound, as shown in the following lemma.

\begin{corollary}
\label{lemma:algorithm-ratio}
For any $\epsilon\le\frac{1}{4}$, the ratio between the memory consumption of Algorithm~\ref{alg:BasicCounting} and the lower bound for additive approximations for \textsc{Basic-Counting} is
$$\frac{\frac{1}{2\epsilon}+2\log W + O(1) }{\lowerbound}=3+o(1).$$
\end{corollary}

Since the proof is very technical, we defer it to Appendix \ref{appendix-bc-optimality}.

%% file: basic-counting-algorithm-pseudo-code.tex
\begin{algorithm}[tb]
\caption{Additive Approximation of Basic Counting}\label{alg:BasicCounting}
\begin{algorithmic}[1]
\State Initialization: $y = 0, b = 0, m = 0, i=0$.
\Function{add}{Bit $x$}
\If {$m = {W\over k} - 1$}
	\State $B = B - b_i$
	\If {$y+x\geq {W\over k}$}
		\State $b_i = 1$
		\State $y = y - {W\over k} + x$
	\Else
		\State $b_i = 0$
		\State $y = y + x$
	\EndIf
	\State $B = B + b_i$
	\State $m = 0$
	\State $i = i+1 \mod k$	
\Else
	\State $y = y + x$
	\State $m = m + 1$
\EndIf
\EndFunction

\Function{Query}{}
\State \Return {${W\over k} \cdot B + y  - \frac{W}{2k}- m\cdot b_i$}
\EndFunction
\end{algorithmic}
\end{algorithm}

%% file: basic_summing.tex
\section{Basic-Summing Problem}\label{sec:basic-summing}

We now consider an extension of \bc{} where elements are non-negative integers:

\begin{definition}{(\bs)}
Given a stream of elements comprising of integers in the range $[\bsrange+1]=\frange\bsrange$, maintain the sum $S$ of the last $W$ elements. 
\end{definition}

\input{summing-lower-bound}

\input{basic-summing-algorithm-analysis}
\subsection{Summing with Small Error}\label{sec:small-error}
Algorithm~\ref{alg:basic-summing} only works for \largeEps{} that satisfies $\frac{W}{k} \ge 1$; otherwise, $\numBlocks$ cannot represent the number of blocks, as blocks cannot be empty.
To complete the picture, we present Algorithm~\ref{alg:basic-summing-small-error} that works for smaller errors.
Intuitively, we keep an array $\bitarray$ of size $\window$, such that every cell represents the number of integer multiples of $\frac{\bsrange\window}{\numBlocks}$ in an arriving item.
Similarly to the above algorithms, we reduce the error by tracking the remainder in a variable $\remainder$, propagating uncounted fractions to the following item. In this case as well, the optimistic approach reduces the error compared with keeping a $\window$-sized array of rounded values for approximating the sum.
Each cell in $b$ needs to represent a value in $\frange{\floor{1+\iblockSize}}$; the remainder $\remainder$ is now a fractional number, represented using $\bsReminderFractionBits$ bits.
When a new item is added, we scale it, add the result to $\remainder$, and update both $\currentBlock$ and the remainder.


\begin{algorithm}[t]
\caption{Additive Approximation for Basic-Summing with Small Error}\label{alg:basic-summing-small-error}
\begin{algorithmic}[1]
\State Initialization: $\remainder = 0, \bitarray = 0, \sumOfBits = 0, \currentBlockIndex=0$.
\Function{\add[\text{element }\inputVariable]}{}
\State $\bsFracInput = \text{Round}_{\bsReminderFractionBits}(\frac{\inputVariable}{R})$
\State $\sumOfBits = \sumOfBits - \currentBlock$
\State $\currentBlock = \floor{\frac{\remainder+\bsFracInput}{W/k}}$
\State $\remainder = \remainder + \bsFracInput - \currentBlock\cdot{\blockSize}$\label{line:small-error-prop}
\State $\sumOfBits = \sumOfBits + \currentBlock$
\State \inc \currentBlockIndex $ \mod W$	
\EndFunction

\Function{\query}{}
\State \Return {$\bsrange \cdotpa{{\blockSize} \cdot \sumOfBits + \remainder  - \halfBlock}$}
\EndFunction

\end{algorithmic}
\end{algorithm}
\normalsize
\begin{theorem} \label{alg:small-error-correctness}
Algorithm~\ref{alg:basic-summing-small-error} provides an $\bserror$-additive approximation for \bs.
\end{theorem}
The proof is delayed to Appendix~\ref{app:small-errors-approx-proof}.
It considers the rounding error generated by representing $\bsFracInput$ using $\bsReminderFractionBits$ bits, and shows that the remainder propagation (Line~\ref{line:small-error-prop}) limits error accumulation.

\begin{theorem}\label{thm:small-eps-mem-requirements}
\smallEpsRestriction{} Algorithm \ref{alg:basic-summing-small-error} requires $W\logp{{1\over 2W\eps}+1}\cdot\smallMultError\le\frac{1}{2\eps}\cdot\smallMultError$ memory bits.
\end{theorem}
The proof is similar to the proof of Theorem~\ref{thm:counting-memory} and therefore deferred to Appendix~\ref{app:small-eps-mem-req-proof}.

We conclude the section by showing that our algorithm is succinct, requiring only $\smallMultError$ times as much memory as the lower bound proved in Theorem~\ref{thm:summing-lower-bound}.
\begin{theorem}
Let $\eps = o(W^{-1})$, and denote $\mathcal{B}\triangleq\window \log \floor{\frac{1}{4\weps} + 1}$. Algorithm~\ref{alg:basic-summing-small-error} provides $\bserror$ additive approximation to \bs{} using $\mathcal{B}\cdot\smallMultError$ memory bits.
\end{theorem}
\begin{proof}
Theorem~\ref{thm:small-eps-mem-requirements} shows that the number of bits our algorithm requires for $\eps = o(W^{-1})$ is $W\logp{{1\over 2W\eps}+1}\cdot\smallMultError
\le \mathcal{B}(1+\frac{2\window}{\mathcal{B}})\cdot\smallMultError
=\mathcal{B}\cdot\smallMultError$.
\end{proof} 

%% file: summing-lower-bound.tex
\subsection{Lower Bound}
We now show that approximating \bs{} to within an additive error of $\bserror$ requires $\Omega(\frac{1}{\eps}+\logw)$ bits for $\eps \ge \frac{1}{2\window}$ and $\Omega(\window\logp{\frac{1}{\weps}})$ bits for $\eps \le \frac{1}{2\window}$.
\begin{lemma}\label{lem:summing-lb-1}
For any $\eps\le\frac{1}{4}$, approximating \bs{} to within an additive error of $\bserror$ requires $\floor\lowerbound$ memory bits.
\end{lemma}
\begin{proof}
The proof of the lemma is very similar to the proof of Theorem~\ref{thm:BasicCountingLB} and is obtained by replacing every set bit with the integer $\bsrange$ in Lemma~\ref{lem:approx-counting-1} and Lemma~\ref{lem:approx-counting-2}.
\end{proof}

Next, we show a lower bound for smaller values of $\eps$.

\begin{lemma}\label{lem:summing-lb-2}
For any $\eps$, approximating \bs{} to within an additive error of $\bserror$ requires at least $\window \log \floor{\frac{1}{4\weps} + 1}$ memory bits.
\end{lemma}
\begin{proof}
Denote $x\triangleq\floor{2\bserror + 1}$ and $C \triangleq \set{n\cdot x\mid n\in\frange{\floor{\frac{1}{2\weps + \bsrange^{-1}}}}}$.
Let $L$ be the language of all $\window$ length strings over the number in $C$, i.e.,
$$L_{R,W,\eps} = \set{\sigma_0\sigma_1\cdots\sigma_{\window-1} | \forall j\in[\window]: \sigma_j\in C}.$$
We show that every two distinct sequences in $L$ must lead the algorithm into distinct configurations implying a lower bound of
$$\ceil {\log |L|} \ge \window\log |C| = \window \log \floor{\frac{1}{2\weps + \bsrange^{-1}}+1}\ge \window \log \floor{\frac{1}{4\weps}+1}$$
bits, where the last inequality follows from the fact that any $\eps < \frac{1}{2\bsrange\window}$ implies exact summing.
Assume, by way of contradiction, that two different words
$$s^1=\sigma_0^1\sigma_1^1\cdots\sigma_{\window-1}^1, s^2=\sigma_0^2\sigma_1^2\cdots\sigma_{\window-1}^2\in L$$
lead the algorithm to the same configuration.
Denote the index of the last letter that differs between $s^1$ and $s^2$
by $t \triangleq max\{\tau\mid \sigma_\tau^1\neq \sigma_\tau^2\}.$
Next, consider the sequences
$s^1\cdot 0^{t-1}$ and
$s^2\cdot 0^{t-1}.$
The algorithm must reach the same configuration after processing these sequences, even though the sum of the last $\window$ elements differ by at least $x=\floor{2\bserror + 1}>2\bserror$.
Therefore, the algorithm's error must be greater than $\bserror$ at least for one of the sequences, in contradiction to the assumption.
\end{proof}

\begin{theorem}\label{thm:summing-lower-bound}
Approximating \bs{} to within an additive error of $\bserror$ requires $\Omega(\frac{1}{\eps}+\logw)$ bits  for $\frac{1}{2\window}\le\eps\le\frac{1}{4}$ and $\Omega(\window\logp{\frac{1}{\weps}})$ bits for $\eps \le \frac{1}{2\window}$.
\end{theorem}
\begin{proof}
Lemma~\ref{lem:summing-lb-1} shows that approximating \bs{} within $\bserror$ requires $$\lowerbound$$ bits for $\frac{1}{2\window}\le\eps\le\frac{1}{4}.$
The same argument used in Lemma~\ref{lemma:algorithm-ratio} shows that this implies $\Omega\parentheses{\frac{1}{\eps}+\logw}$ bits lower bound for any $\eps\ge\frac{1}{2\window}$.
For $\epsilon<\frac{1}{2\window}$ such that $\epsilon=\Theta(W^{-1})$, approximating \bs{} within $\bserror$ implies a $\frac{\bsrange}{2}$-additive approximation and therefore the $\Omega\parentheses{\frac{1}{\eps}+\logw}$ bound holds.
For $\eps = o\parentheses{\frac{1}{\window}}$, we use Lemma~\ref{lem:summing-lb-2}, which implies a lower bound of $\window \log \floor{\frac{1}{4\weps} + 1}= \Omega\parentheses{\smallEpsLowerbound}$~memory~bits.
\end{proof}
An immediate corollary of Theorem~\ref{thm:summing-lower-bound} is that any exact algorithm for \bs{} requires at least $\Omega(W\log\bsrange)$ bits, i.e., the naive solution of maintaining a $\window$-sized array of the elements in the window, encoding each using $\ceil{\logp{\bsrange+1}}$ bits, is optimal (for exact \bs).
Finally, we extend the results to randomized algorithms, where the proof appears in Appendix~\ref{apx:bsRandomizedLB}.
\begin{theorem}\label{thm:bsRandomizedLB}
For any fixed $\delta\in[0,1/2)$, any randomized Monte Carlo algorithm that gives a $W\epsilon$ approximation to \bs{} with a probability of at least $1-\delta$ requires $\Omega(\frac{1}{\eps}+\logw)$ bits  for $\frac{1}{2\window}\le\eps\le\frac{1}{4}$ and $\Omega(\window\logp{\frac{1}{\weps}})$ bits for $\eps \le \frac{1}{2\window}$.
	Notice that the $\delta=0$ case applies to Las Vegas algorithms.
\end{theorem}

%% file: basic-summing-algorithm-analysis.tex
\subsection{Upper Bound}
We show that our \bc{} algorithm can be adapted to this problem with only a small memory overhead such that the algorithm's state size remains independent of $\bsrange$. We first present the extension of the algorithm for the  \largeEps{} case.
In Section~\ref{sec:small-error} we complete the picture by giving an alternative algorithm for smaller values of $\eps$.
Intuitively, we ``scale'' the algorithm by dividing each added element by $\bsrange$ and rounding the result.
In order to keep the sum of elements not yet accounted for in $\bitarray$, $\remainder$ is now maintained as a fixed-point variable rather than an integer.
Ideally, the fractional value of the remainder $\remainder$ should allow exact representation of $\set{0,1/\bsrange,\ldots,1-1/\bsrange}$, and therefore requires $\log\bsrange$ bits.
When the range $\bsrange$ is ``large'', or more precisely $\bsrange=\omega(\epsilon^{-1})$, we save space by storing the fractional value of $\remainder$ using less than $\log\bsrange$ bits, which inflicts a rounding error.
That is, we keep $\remainder$ using 
$\ceil{\logp{2\blockSize}}+\bsReminderFractionBits$ bits.
Similarly to our \bc{} algorithm, $\ceil{\logp{2\blockSize}}$ bits are used to store the integral part of $\remainder$.
The additional $\bsReminderFractionBits$ bits are used for the fractional value of $\remainder$. The value of $\bsReminderFractionBits$ is determined later.

In order to keep the total error bounded, we compensate for the rounding error by using smaller block sizes, which are derived from the number of blocks $\numBlocks$, determined in \eqref{eq:numBlocks}.
Our algorithm keeps the following variables:
\begin{description}
	
\footnotesize
\item [$\bitarray$ -] a bit-array of size $\numBlocks$.
\item [$\remainder$ -] a counter for the sum of elements which is not yet accounted for in $\bitarray$.
\item [$\currentBlockIndex$ -] the index of the ``oldest'' block in $\bitarray$.
\item [$\sumOfBits$ -] the sum of all bits in $\bitarray$.
\item [$\blockOffset$ -] a counter for the current offset within the block.
\normalsize
\end{description}
Our \bs{} algorithm is presented in Algorithm~\ref{alg:basic-summing}.
We use $\text{Round}_{\bsReminderFractionBits}(z)$ for some $z\in[0,1]$ to denote rounding of $z$ to the nearest value $\widetilde{z}$ such that $2^{\bsReminderFractionBits}\widetilde{z}$ is an integer.

\begin{algorithm}[h]
\caption{Additive Approximation for Basic-Summing}\label{alg:basic-summing}
\begin{algorithmic}[1]
\State Initialization: $\remainder = 0, \bitarray = 0, \sumOfBits = 0, \currentBlockIndex=0, \blockOffset=0$.
\Function{\add[\text{element }\inputVariable]}{}
\State $\bsFracInput = \text{Round}_{\bsReminderFractionBits}(\frac{\inputVariable}{R})$ \label{line:rounding}
\If {$\blockOffset = {\blockSize} - 1$}\label{line:end-of-block}
	\State $\sumOfBits = \sumOfBits - \currentBlock$
	\If {$\remainder+\bsFracInput\geq {\blockSize}$}
		\State $\currentBlock = 1$
		\State $\remainder = \remainder - {\blockSize} + \bsFracInput$
	\Else
		\State $\currentBlock = 0$
		\State $\remainder = \remainder + \bsFracInput$
	\EndIf
	\State $\sumOfBits = \sumOfBits + \currentBlock$
	\State $\blockOffset = 0$
	\State \inc \currentBlockIndex $ \mod \numBlocks$	
\Else
	\State $\remainder = \remainder + \bsFracInput$
	\State \inc \blockOffset
\EndIf
\EndFunction

\Function{\query}{}
\State \Return {$\bsrange \cdotpa{{\blockSize} \cdot \sumOfBits + \remainder  - \halfBlock - \blockOffset \cdot \currentBlock}$}
\EndFunction

\end{algorithmic}
\end{algorithm}
\normalsize

\begin{theorem} \label{thm:approximation-ratio}
\largeEpsRestriction{} Algorithm~\ref{alg:basic-summing} provides an $\bserror$-additive approximation for \bs.
\end{theorem}

Theorem \ref{thm:approximation-ratio} shows that for any \largeEps{}, by choosing $\bsReminderFractionBits\triangleq \ceil{\logp{\eps^{-1}\log W}}$ and the number of blocks to be $\numBlocks=\ceil{ \frac{1}{2\eps -2^{-\bsReminderFractionBits}}}$, our algorithm estimates $S$ with an additive error of $\bserror$.
The proof of Theorem~\ref{thm:approximation-ratio} can be found in Appendix~\ref{sec:proof-basic-summing}.
The following theorem analyzes the memory requirements of our algorithm.
\begin{theorem}\label{thm:bs-memory-requirements}
\largeEpsRestriction{} Algorithm~\ref{alg:basic-summing} requires $\parentheses{2\logw + \frac{1}{2\epsilon}}\smallMultError{}$ memory bits.
\end{theorem}


The proof is similar to the proof of Theorem~\ref{thm:counting-memory} and is therefore deferred to Appendix~\ref{app:bs-memory-requirements-proof}.

%% file: discussion.tex
\section{Discussion}
\label{sec:discussion}

In this paper, we have investigated additive approximations for the \bc{} and \bs{} problems.
For both cases, we have provided space efficient algorithms.
Further, we have proved the first lower bound for additive approximations for the \bc{} problem, and showed that our algorithm achieves this bound,
and is hence optimal.
In the case of \bs{}, whenever \largeEps, the same lower bound as in the \bc{} problems still holds and so our approximation algorithm for this domain is
optimal up to a small factor.
For other values of $\eps$, we have shown an improved lower bound and a corresponding succinct approximation algorithm.

In the future, we would like to study lower and upper bounds for additive approximations for several related problems.
These include, e.g., approximating the sliding window sum of weights for each item in a stream of (item, weight) tuples. 
Further, we intend to explore applying additive approximations in the case of multiple streams.
Obviously, one can allocate a separate counter for each stream, thereby multiplying the space complexity by the number of concurrent streams.
However, it was shown in~\cite{GibbonsT02} that for the case of multiplicative approximations, there is a more space efficient solution.
We hope to show a similar result for additive approximations.
\vspace{-0.1cm}

%% file: bibliography.tex
\nocite{Simpson}

%% file: basic_counting_lemma.tex
\section{Proof of Lemma~\ref{lemma:algorithm-ratio}}\label{appendix-bc-optimality}

\begin{proof}
We make a simple case analysis to show that the above holds:
\begin{description}
\item[$\epsilon \geq \frac{\log ^{-1} W - W^{-1}}{2}$:]
In this case the maximum in the denominator is $\log W$. Hence
\begin{align*}
\frac{\frac{1}{2\epsilon}+2\log W+O(1)}{\lowerbound} 
&\leq\frac{\frac{1}{\log ^{-1} W - W^{-1}}+2\log W+O(1)}{\log W} 
=2+\frac{\frac{1}{\log ^{-1} W - W^{-1}}+O(1)}{\log W}\\
&=2+\frac{W\log W+O(W)}{\log W (W-\log W)}
=3+\frac{\log ^2 W+O(W)}{\log W (W-\log W)}\\
&=3+\frac{\log W}{W-\log W}+\frac{O(1)}{\log W}
=3+o(1)
\end{align*}
\item[$\epsilon < \frac{\log ^{-1} W - W^{-1}}{2}$:] In this case, the maximum in the denominator is $\frac{1}{2\epsilon+W^{-1}}$. Hence
\begin{align*}
\frac{\frac{1}{2\epsilon}+2\log W+O(1)}{\lowerbound}
&=\frac{\frac{1}{2\epsilon}+2\log W+O(1)}{ \frac{1}{2\epsilon+W^{-1}}}\\
&< \frac{2\epsilon+W^{-1}}{2\epsilon}+2\log W(2\epsilon+W^{-1}) + \frac{O(1)}{\log W}\\
&= 1+\frac{W^{-1}}{2\eps}+4\epsilon\log W +2\frac{\log W}{W} + \frac{O(1)}{\log W}\\
&=1+\frac{W^{-1}}{2\epsilon}+4\epsilon\log W+o(1).
\end{align*}
An analysis of the function
$$g(W) = 1+\frac{W^{-1}}{2\epsilon}+4\epsilon\log W$$
reveals that it has a single extremum for positive $W$s, located at $W=\frac{1}{8\epsilon^2}$.
We check the value of the function at the extremum and at its boundaries:
\begin{description}
\item [$W = \frac{1}{2\epsilon}$:]
This is a boundary because $\epsilon<\frac{1}{2W}$ means that the algorithm is required to give an exact counting of the number of ones.
On this boundary we get
$$
g\left(\frac{1}{2\epsilon}\right)
=1+\frac{2\epsilon}{2\epsilon}-4\epsilon\log (2\epsilon)=2+{2\over W}\log W=2+o(1).
$$

\item [The extremum $W=\frac{1}{8\epsilon^2}$:]
\begin{align*}
g\left(\frac{1}{8\epsilon^2}\right)&=1+\frac{8\epsilon^2}{2\epsilon}-4\epsilon\log (8\epsilon^2) + o(1)=1+{4\over \sqrt{8W}}\cdot(1+\log W)+o(1)\\
&=1+o(1)
\end{align*}
\end{description}
\end{description}
\end{proof}

%% file: randomized_bs.tex
\section{Proof of Theorem~\ref{thm:bsRandomizedLB}}\label{apx:bsRandomizedLB}
\begin{proof}
The case for $\Omega(\frac{1}{\eps}+\logw)$ bits when $\frac{1}{2\window}\le\eps\le\frac{1}{4}$ is similar to Theorem~\ref{thm:RandomizedBasicCountingLB}. The bound is obtained by replacing every set bit with the integer $\bsrange$ in the corresponding languages.
In case $\eps<\frac{1}{2W}$, we consider the padded language obtained by adding $W$ $0$'s to every word in the language defined in Lemma~\ref{lem:summing-lb-2}.
Further, we define the distribution $\mathbf p$ to be uniform over
$$L = L_{R,W,\eps}\cdot 0^W.$$
Similarly to Theorem~\ref{thm:RandomizedBasicCountingLB}, Lemma~\ref{lem:summing-lb-2} implies that for (at least) $(1-2\delta)$ fraction of the inputs in $L$, the algorithm must reach a distinct state after reading the first $W$ integers. Therefore, the derived lower bound according to Yao's Minimax principle~\cite{minimax} is:
$$B_{MC}\ge \frac{1}{2}\logp{(1-2\delta)\cdot|L|}\ge\frac{1}{2}\window \log \floor{\frac{1}{4\weps}+1}+\frac{1}{2}\logp{1-2\delta}=
\Omega\parentheses{\window\logp{\frac{1}{\weps}}}.$$ 
\end{proof}

%% file: basic_summing_error.tex

\section{Proof of Theorem~\ref{thm:approximation-ratio}}
\label{sec:proof-basic-summing}

\begin{proof}
First, let us introduce some notations used in the proof. Element $j$ is marked $\inputVariable_j$ or $\bsFracInput_j$ after we divide it by $\bsrange$ and round it.
Assume that the index of the last element is $\window+\blockOffset$, where $\inputVariable_\window$ is the last element of a block.3
When we refer to the block index $\currentBlockIndex$, we refer to its value after $\window+\blockOffset$ elements have been processed.
We denote $\remainder_j$ the value of $\remainder$ \emph{after} adding $\inputVariable_j$.

The setting for the proof is given in Figure \ref{fig:basic-summing}.
Our goal is to approximate
\begin{equation}
\bsAnalysisTarget \triangleq \sum_{j=\blockOffset+1}^{\window+\blockOffset} \inputVariable_j \label{eq:bs-to-approximate}
\end{equation}

\begin{figure}[t]
\centering
\includegraphics[scale=0.28]{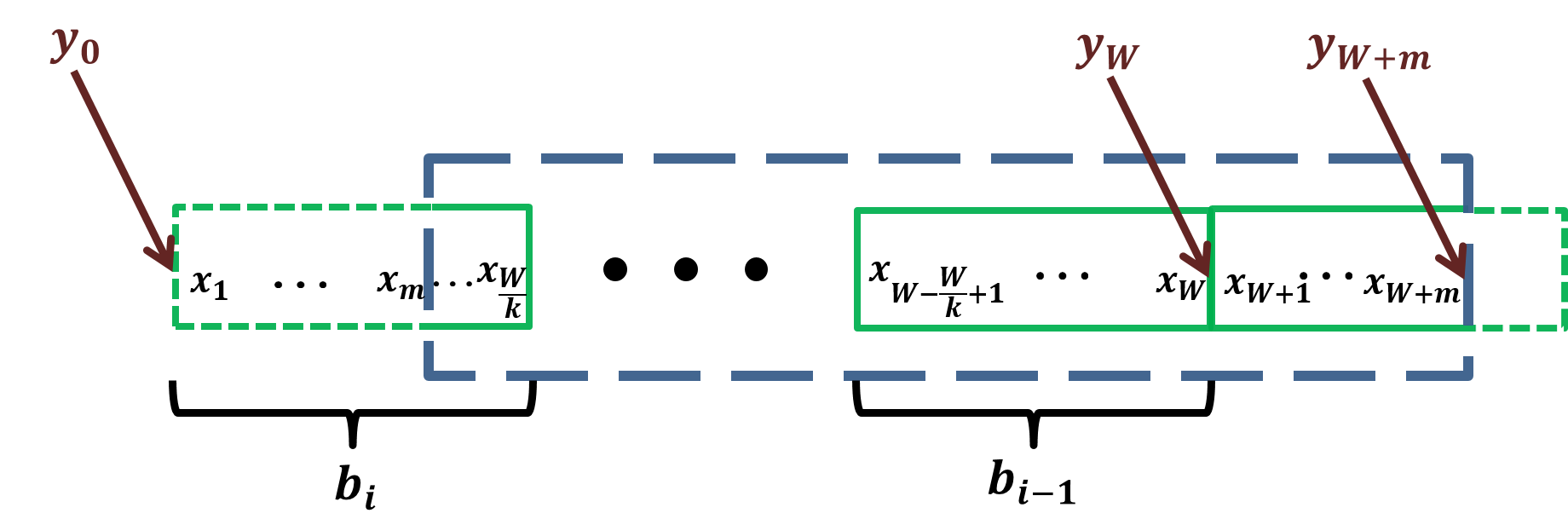}
\caption{The setting for the proof of Theorem~\ref{thm:approximation-ratio}.}
\label{fig:basic-summing}
\end{figure}

Algorithm~\ref{alg:basic-summing} uses the following approximation to answer this query:
\begin{align}
\bsAnalysisEstimator
&= \bsrange \cdotpa{\blockSize \cdot \sumOfBits + \remainder_{\window+\blockOffset} - \halfBlock -\blockOffset\cdot \currentBlock }\notag\\
&= \bsrange \cdotpa{\blockSize \cdot \sumOfBits + \remainder_{\window} + \sum_{j=\window+1}^{\window+m} \bsFracInput_j - \halfBlock -\blockOffset\cdot \currentBlock} \label{eq:bs-approximation1}
 \end{align}
At the end of block $j$, if $\remainder$ is decreased by $\blockSize$, then $\bitarray_j$ is set and will not be cleared before time $\window+\blockOffset$.
Therefore $$\blockSize \cdot \sumOfBits + \remainder_{\window} = \remainder_0 + \sum_{j=1}^{\window} \bsFracInput_j.$$

Substituting $\frac{W}{k} \cdot B + y_W$ in \eqref{eq:bs-approximation1} we get
\begin{align}
\bsAnalysisEstimator
&= \bsrange \cdotpa{\remainder_0 + \sum_{j=1}^{\window} \bsFracInput_j  + \sum_{j=\window+1}^{\window+\blockOffset} \bsFracInput_j - \halfBlock-\blockOffset\cdot \currentBlock}\notag\\
&= \bsrange \cdotpa{\remainder_0 + \sum_{j=1}^{\blockOffset} \bsFracInput_j  + \sum_{j=\blockOffset+1}^{\window+\blockOffset} \bsFracInput_j - \halfBlock-\blockOffset\cdot \currentBlock}\label{eq:bs-post-sub}.
\end{align}
Denote the rounding error over the entire window
\begin{equation*}
\bsRoundingError\triangleq\sum_{j=\blockOffset+1}^{\window+\blockOffset} \inputVariable_j - \bsrange\cdot{\sum_{j=\blockOffset+1}^{\window+\blockOffset} \bsFracInput_j}.
\end{equation*}
From line \ref{line:rounding}, we know that
$\bsFracInput_j =\text{Round}_{\bsReminderFractionBits}\parentheses{\frac{\inputVariable_j}{R}},$
while $\abs{\frac{\inputVariable_j}{\bsrange}- \text{Round}_{\bsReminderFractionBits}\parentheses{\frac{\inputVariable_j}{R}}}\le 2^{-1-\bsReminderFractionBits},$ and thus:
\begin{equation}
-2^{-1-\bsReminderFractionBits}\cdot\bsrange\window\le\bsRoundingError\le2^{-1-\bsReminderFractionBits}\cdot\bsrange\window\label{eq:bs-rounding-error}
\end{equation}
if $\bsReminderFractionBits<\log\bsrange$, or $0$ otherwise (as we then store $\frac{\inputVariable_j}{R}$ accurately).
Thus, we can write \eqref{eq:bs-post-sub} as:
\begin{equation*}
\bsAnalysisEstimator=
\sum_{j=\blockOffset+1}^{\window+\blockOffset} \inputVariable_j + \bsRoundingError + \bsrange \cdotpa{\remainder_0 + \sum_{j=1}^{\blockOffset} \bsFracInput_j   - \halfBlock-\blockOffset\cdot \currentBlock}.
\end{equation*}


Plugging \eqref{eq:bs-to-approximate}, we get
\begin{equation*}
\bsAnalysisEstimator=
\bsAnalysisTarget + \bsRoundingError + \bsrange \cdotpa{\remainder_0 + \sum_{j=1}^{\blockOffset} \bsFracInput_j   - \halfBlock-\blockOffset\cdot \currentBlock}.
\end{equation*}
Therefore, the error is
\begin{equation}
\bsAnalysisError = \bsRoundingError + \bsrange \cdotpa{\remainder_0 + \sum_{j=1}^{\blockOffset} \bsFracInput_j   - \halfBlock-\blockOffset\cdot \currentBlock}\label{eq:bs-est-error}.
\end{equation}

We consider two cases:
\begin{itemize}
\item $\currentBlock = 1$: This means that $\remainder$ has crossed the threshold by time $\blockSize$, i.e.,
$$\remainder_0 + \sum_{j=1}^{\blockSize}\bsFracInput_j \geq \blockSize
\mathrm{\ and\ equivalently\ }
\remainder_0 + \sum_{j=1}^{\blockOffset}\bsFracInput_j \geq \blockSize - \sum_{j=\blockOffset+1}^{\blockSize}\bsFracInput_j.$$
Going back to \eqref{eq:bs-est-error}, we get on one side
\begin{align}
\bsAnalysisError &= \bsRoundingError + \bsrange \cdotpa{\remainder_0 + \sum_{j=1}^{\blockOffset} \bsFracInput_j   - \halfBlock-\blockOffset\cdot \currentBlock} \notag \\
&\ge \bsRoundingError + \bsrange \cdotpa{\blockSize - \sum_{j=\blockOffset+1}^{\blockSize}\bsFracInput_j   - \halfBlock-\blockOffset} \notag\\
&\ge \bsRoundingError + \bsrange \cdotpa{\halfBlock - \left(\sum_{j=\blockOffset+1}^{\blockSize}1\right) - \blockOffset} \notag\\
&\ge \bsRoundingError - \bsrange \cdot\halfBlock \label{eq:bs-error1}
\end{align}

To bound the error from above we use the fact that the value of $\remainder$ at the end of a block never exceeds $\threshold$.
This can be shown by induction, as $\remainder$ is increased at most $\blockSize$ times during one block, and then reduced by $\threshold$ if it exceeds the block size.
Therefore, \eqref{eq:bs-est-error} can be bounded as follows:
\begin{align}
\bsAnalysisError &= \bsRoundingError + \bsrange \cdotpa{\remainder_0 + \sum_{j=1}^{\blockOffset} \bsFracInput_j - \halfBlock-\blockOffset\cdot \currentBlock}\notag\\
&\le \bsRoundingError + \bsrange \cdotpa{\remainder_0- \halfBlock}\le \bsRoundingError + \bsrange \cdot\halfBlock\label{eq:bs-error2}.
\end{align}

\item $\currentBlock = 0$: Similarly, this means that $\remainder$ was smaller than the threshold at the end of block $\currentBlockIndex$. Hence
$\remainder_0 + \sum_{j=1}^{\blockSize}\bsFracInput_j \leq \blockSize - 1$, or equivalently
    $$\remainder_0 + \sum_{j=1}^{\blockOffset}\bsFracInput \leq \blockSize - \sum_{j=\blockOffset+1}^{\blockSize}\bsFracInput_j - 1.$$
    Thus, the upper bound is
        \begin{align}
    	\bsAnalysisError &= \bsRoundingError + \bsrange \cdotpa{\remainder_0 + \sum_{j=1}^{\blockOffset} \bsFracInput_j   - \halfBlock}\notag \\
    	&\le\bsRoundingError + \bsrange \cdotpa{\blockSize - \sum_{j=\blockOffset+1}^{\blockSize}\bsFracInput_j - 1 - \halfBlock} 
        \le \bsRoundingError + \bsrange \cdot\halfBlock\label{eq:bs-error4}.
        \end{align}
        Our error is then bounded from below by
    \begin{align}
\bsAnalysisError &= \bsRoundingError + \bsrange \cdotpa{\remainder_0 + \sum_{j=1}^{\blockOffset} \bsFracInput_j   - \halfBlock} 
\ge\bsRoundingError - \bsrange \cdot \halfBlock\label{eq:bs-error3}.
    \end{align}

\end{itemize}
We need to bound \eqref{eq:bs-error1},\eqref{eq:bs-error2},\eqref{eq:bs-error4} and \eqref{eq:bs-error3} by $\bserror$. Using \eqref{eq:bs-rounding-error}, it is enough to require
\begin{equation*}
2^{-1-\bsReminderFractionBits}+\frac{1}{2k}\le \epsilon.
\end{equation*}
Therefore, we choose the number of blocks to be
\begin{equation}
\numBlocks\triangleq\ceil{ \frac{1}{2\eps -2^{-\bsReminderFractionBits}}}\label{eq:numBlocks}.
\end{equation}
and thus provide an estimation that is an $\bserror$-additive approximation for \bs.

Finally, we choose the number of bits representing the fractional value of $\remainder$ to be $$\bsReminderFractionBits\triangleq \ceil{\logp{\eps^{-1}\log W}}
\ge\logp{\eps^{-1}\log W}$$.

Therefore, the number of blocks is
$$
\numBlocks=\ceil{ \frac{1}{2\eps -2^{-\bsReminderFractionBits}}}
\le\ceil{\frac{1}{2\eps-{\eps\over\log W}}}
=\ceil{\frac{\logw}{2\eps(\logw-1)}}
$$
To have $k\le W$ it suffices to require
$\eps^{-1} \le 2W\left(1-\frac{1}{\logw}\right)$.

\end{proof}

%% file: basic-summing-memory.tex
\section{Proof of Theorem~\ref{thm:bs-memory-requirements}}\label{app:bs-memory-requirements-proof}
\begin{proof}
	We represent $\remainder$ using $\ceil{2+\logweps+\bsReminderFractionBits}$ bits, $\blockOffset$ using
	$\ceil{1+\logweps}$ bits and $\bitarray$ using $\numBlocks=\ceil{ \frac{1}{2\eps -2^{-\bsReminderFractionBits}}}$ bits.
	Additionally, $\currentBlockIndex$ requires $\ceil{\log \numBlocks}$ bits, and $\sumOfBits$ another $\ceil{\logp{\numBlocks+1}}$ bits.
	Overall the number of bits required is
	\begin{align*}
	& \numBlocks + \ceil { \logweps + \bsReminderFractionBits} + \ceil\logweps + \ceil{\log \numBlocks} + \ceil{\logp{\numBlocks+1}}\\
	&=\numBlocks + 2\logweps +  \bsReminderFractionBits + 2\log \numBlocks +O(1).
	\end{align*}
	The number of blocks is  \footnotesize
	\begin{align*}
	\numBlocks&=\ceil{ \frac{1}{2\eps -2^{-\bsReminderFractionBits}}}
	\le 1 + \frac{1}{2\eps-{\eps\over\log W}}
	=1+\frac{1}{2}\eps^{-1}\cdot{2\logw\over2\logw-1}\\
	&=1+\frac{1}{2}\eps^{-1}\cdot\left(1+{1\over 2\logw-1}\right)
	=1+\frac{1}{2}\eps^{-1}\smallMultError.
	\end{align*}\normalsize
	Thus, the space consumption becomes
		\footnotesize
	\begin{align*}
	&\numBlocks + 2\logweps +  \bsReminderFractionBits + 2\log \numBlocks +O(1)\\
	&=\left({\eps^{-1}\over 2}+2\logw-2\log\eps^{-1}+\ceil{\logp{\eps^{-1}\log W}}
	+2\logp{1+{\eps^{-1}\over 2}}\right)\cdot\smallMultError\\
	&=\left({\eps^{-1}\over 2}+2\logw-\log\eps^{-1}+2\logp{1+{\eps^{-1}\over 2}}\right)\smallMultError\\
	&\le\left({\eps^{-1}\over 2}+2\logw+2\logp{1+{\eps^{-1}\over 2}}\right)\smallMultError.
	\end{align*} \normalsize
	There are two cases. If $\frac{\eps^{-1}}{2}\le\logw-1$, the space is less than or equal to\footnotesize
	$$\left({\eps^{-1}\over 2}+2\logw+2\log\logw\right)\smallMultError
	=\parentheses{2\logw + \frac{1}{2\epsilon}}\smallMultError{}.$$\normalsize
	Otherwise,\footnotesize
	\begin{align*}
	&\left({\eps^{-1}\over 2}+2\logw+\eps^{-1}\cdot\frac{2}{\eps^{-1}}\logp{1+{\eps^{-1}\over 2}}\right)
	\smallMultError\\
	\le&\left({\eps^{-1}\over 2}+2\logw+\eps^{-1}{\log\logw\over\left(\logw-1\right)}\right)\smallMultError
	=\parentheses{2\logw + \frac{1}{2\epsilon}}\smallMultError{}.
	\end{align*}\normalsize
\end{proof}

%% file: small-eps-correctness.tex
\section{Proof of Theorem~\ref{alg:small-error-correctness}}\label{app:small-errors-approx-proof}
\begin{proof}
The notations for this proof are the same as in Appendix~\ref{sec:proof-basic-summing}. In this case all "blocks" are of size $1$ so $\blockOffset$ is always $0$.

If at time $j$ $\remainder$ is decreased by $\alpha\cdot\blockSize$, then $\bitarray_j$ is set to $\alpha$ and will not change before time $\window$.
Therefore $$\blockSize \cdot \sumOfBits + \remainder_{\window} = \remainder_0 + \sum_{j=1}^{\window} \bsFracInput_j.$$
Substituting $\frac{W}{k} \cdot B + y_W$ in \eqref{eq:bs-approximation1} we get
\begin{align}
\bsAnalysisEstimator
&= \bsrange \cdotpa{\remainder_0 + \sum_{j=1}^{\window} \bsFracInput_j - \halfBlock}\notag\\
\label{eq:small-error-bs-post-sub}.
\end{align}

Plugging \eqref{eq:bs-to-approximate} and taking into consideration the rounding error, we get
\begin{equation*}
\bsAnalysisEstimator=
\bsAnalysisTarget + \bsRoundingError + \bsrange \cdotpa{\remainder_0 - \halfBlock}.
\end{equation*}
Therefore, the error is
\begin{equation}
\bsAnalysisEstimator-\bsAnalysisTarget =
\bsRoundingError + \bsrange \cdotpa{\remainder_0 - \halfBlock}.\label{eq:small-error-bs-est-error}
\end{equation}

To bound the error we use the fact that $0\le\remainder<\threshold$ at the end of a block.
Therefore, \eqref{eq:small-error-bs-est-error} can be bounded as follows:
$$\bsRoundingError - \bsrange \cdot\halfBlock \le \bsAnalysisError<\bsRoundingError + \bsrange \cdot\halfBlock$$

$k$ remains  
$$\numBlocks=\ceil{ \frac{1}{2\eps -2^{-\bsReminderFractionBits}}}$$
and the error is bounded by $\pm RW\epsilon$, similarly to Theorem~\ref{thm:approximation-ratio}.
\end{proof}

%% file: small-eps-summing-memory-requirements.tex
\section{Proof of Theorem~\ref{thm:small-eps-mem-requirements}}\label{app:small-eps-mem-req-proof}
\begin{proof}
We represent $\remainder$ using $\bsReminderFractionBits$ bits and $\bitarray$ using $W\ceil{\logp{\iblockSize+1}}$ bits.
Additionally, $\currentBlockIndex$ requires $\ceil{\log W}$ bits, and $\sumOfBits$ another $\ceil{\logp{\numBlocks+1}}$ bits.
Overall the number of bits required is
$$\bsReminderFractionBits +  W\ceil{\log\left({\frac{k}{W}+1}\right)} + \ceil{\log W} + \ceil{\logp{\numBlocks+1}}
=\left(W\logp{{k\over W}+1} + \bsReminderFractionBits+\log k \right)\smallMultError.$$
We choose the number of bits representing the fractional value of $\remainder$ to be $$\bsReminderFractionBits\triangleq \ceil{\logp{\eps^{-1} W}}
\ge\logp{\eps^{-1} W}.$$
This way, $k$ becomes
$$\ceil{ \frac{1}{2\eps -2^{-\bsReminderFractionBits}}}
\le\ceil{\frac{1}{2\eps -{\eps\over W}}}
=\ceil{\frac{1}{2\eps}\left(\frac{2W}{2W-1}\right)}
\le 1+\frac{1}{2}\eps^{-1}\smallMultError.$$
Thus, the space consumption is\footnotesize
\begin{align*}
&\parentheses{W\logp{{k\over W}+1} + \bsReminderFractionBits+\log k }\smallMultError\\
=&\parentheses{W\logp{{k\over W}+1} + \logp{\eps^{-1}W}+\log {k\over W}+\logw}\smallMultError\\
=&\parentheses{W\logp{{k\over W}+1}+ \log\eps^{-1}}\smallMultError.
\end{align*}
\normalsize
Recall that $k\ge{\eps^{-1}\over 2}$ and hence $\eps^{-1}\le 2k$. Overall, space becomes
\footnotesize
\begin{align*}
&\parentheses{W\logp{{k\over W}+1}+ \log\eps^{-1}}\smallMultError
\le \parentheses{W\logp{{k\over W}+1}+ \log k}\smallMultError\\
&=\parentheses{W\logp{{k\over W}+1}+ \log {k\over W}+\logw}\smallMultError
=\parentheses{W\logp{\left(2W\eps\right)^{-1}+1}}\smallMultError\\
&={1\over2\eps}\parentheses{\logp{\left(2W\eps\right)^{-1}+1}\over \left(2W\eps\right)^{-1}}\smallMultError
<\frac{1}{2\eps}\cdot\smallMultError.
\end{align*}\normalsize
\end{proof}

%% file: main.bbl
\begin{thebibliography}{}

\end{thebibliography}


\begin{thebibliography}{50}
	
	\bibitem{lambdaCounter} Lee, L. K. and Ting, H. F. A Simpler and More Efficient Deterministic Scheme for Finding Frequent Items over Sliding Windows, Proceedings of the Twenty-fifth ACM SIGMOD-SIGACT-SIGART Symposium on Principles of Database Systems (PODS) 2006.
	
	
\bibitem{LIS} Albert, M.H., Golynski, A., Hamel, A.M., Lopez-Ortiz, A., Rao, S.S., Safari, M.A.: Longest
increasing subsequences in sliding windows. Theoretical Comp. Sci. 321(2), 405-414
(2004)

\bibitem{Quantitiles} Arasu, A., Manku, G.S.: Approximate counts and quantiles over sliding windows. In: Pro-
ceedings of the 23rd ACM SIGACT-SIGMOD-SIGART Symposium on Principles of Database
Systems, (PODS) 2004.

\bibitem{BabcockDMO03} Babcock, B., Datar, M., Motwani, R., O'Callaghan, L.: Maintaining variance and k-medians
over data stream windows. In: Neven, F., Beeri, C., Milo, T. (eds.) Proceedings of the Twenty-
Second ACM SIGACT-SIGMOD-SIGART Symposium on Principles of Database Systems,
June 9-12, 2003, San Diego, CA, USA. pp. 234-243. ACM (2003)

\bibitem{HHINFOCOM} Ben Basat, R., Einziger, G., Friedman, R., Kassner, Y.: Heavy Hitters in Streams and Sliding Windows. In: Proc. IEEE INFOCOM, 2016.


\bibitem{L2HeavyHitters} Braverman, V., Gelles, R., Ostrovsky, R.: How to catch l2-heavy-hitters on sliding windows.
Theoretical Computer Science 554, 82-94 (2014)

\bibitem{SmoothHistograms} Braverman, V., Ostrovsky, R.: Smooth histograms for sliding windows. In: Foundations of
Computer Science, 2007. FOCS'07. 48th Annual IEEE Symposium on. pp. 283-293.

\bibitem{CohenS06} Cohen, E., Strauss, M.J.: Maintaining time-decaying stream aggregates. J. Algorithms 59(1),
19-36 (2006)

\bibitem{DistributedAggregates} Cormode, G., Yi, K.: Tracking distributed aggregates over time-based sliding windows. In:
Scientific and Statistical Database Management. pp. 416-430. Springer (2012)

\bibitem{WeightedMatching} Crouch, M., Stubbs, D.S.: Improved streaming algorithms for weighted matching, via un-
weighted matching. In: Approximation, Randomization, and Combinatorial Optimization. Al-
gorithms and Techniques, APPROX/RANDOM 2014, September 4-6, 2014, Barcelona, Spain.
pp. 96-104 (2014), http://dx.doi.org/10.4230/LIPIcs.APPROX-RANDOM.2014.96

\bibitem{Graphs} Crouch, M.S., McGregor, A., Stubbs, D.: Dynamic graphs in the sliding-window model. In:
Algorithms-ESA 2013, pp. 337-348. Springer (2013)

\bibitem{DatarGIM02} Datar, M., Gionis, A., Indyk, P., Motwani, R.: Maintaining stream statistics over sliding
windows. SIAM J. Comput. 31(6), 1794-1813 (2002)

\bibitem{GibbonsT02} Gibbons, P.B., Tirthapura, S.: Distributed streams algorithms for sliding windows. In: SPAA.
pp. 63-72 (2002)

\bibitem{HeavyHitters} Hung, R., Ting, H.: Finding heavy hitters over the sliding window of a weighted data stream.
In: Laber, E., Bornstein, C., Nogueira, L., Faria, L. (eds.) LATIN 2008: Theoretical Informatics,
LNCS
, vol. 4957, pp. 699-710. Springer Berlin Heidelberg
(2008)

\bibitem{LeeT06} Lee, L., Ting, H.F.: Maintaining significant stream statistics over sliding windows. In: Proceedings of the Seventeenth Annual ACM-SIAM Symposium on Discrete Algorithms, SODA 2006, Miami, Florida, USA, January 22-26, 2006. pp. 724-732. ACM Press (2006)

\bibitem{SlidingBloomInfocom} Liu, Y., Chen, W., Guan, Y.: Near-optimal approximate membership query over time-decaying
windows. In: INFOCOM, 2013 Proceedings IEEE. pp. 1447-1455 (April 2013)

\bibitem{TopKSlidingWindow} Mouratidis, K., Bakiras, S., Papadias, D.: Continuous monitoring of top-k queries over sliding windows. In: Proceedings of the 2006 ACM SIGMOD International Conference on Management of Data. pp. 635-646. SIGMOD '06, ACM, New York, NY, USA (2006)

\bibitem{SlidingBloomFilter} Naor, M., Yogev, E.: Sliding bloom filters. In: Cai, L., Cheng, S.W., Lam, T.W. (eds.)
Algorithms and Computation, Lecture Notes in Computer Science, vol. 8283, pp. 513-523. Springer Berlin Heidelberg (2013)

\bibitem{TopK2} Pripuzic, K., Zarko, I.P., Aberer, K.: Time- and space-efficient sliding window top-k query processing. ACM Trans. Database Syst. 40(1), 1:1-1:44 (Mar 2015)

\bibitem{duplicate} Shen, H., Zhang, Y.: Improved approximate detection of duplicates for data streams over sliding windows. Journal of Computer Science and Technology 23(6), 973-987 (2008)

\bibitem{TopKPairsSlidingWindows} Shen, Z., Cheema, M., Lin, X., Zhang, W., Wang, H.: Efficiently monitoring top-k pairs over
sliding windows. In: Data Engineering (ICDE), 2012 IEEE 28th International Conference on.
pp. 798-809 (April 2012)

\bibitem{DistinctSlidingWindow} Zhang, W., Zhang, Y., Cheema, M.A., Lin, X.: Counting distinct objects over sliding windows. In: Proceedings of the Twenty-First Australasian Conference on Database Technologies
- Volume 104. pp. 75-84. ADC '10, Australian Computer Society, Inc., Darlinghurst, Australia,
Australia (2010).



\bibitem{minimax} Yao AC. Probabilistic computations: Toward a unified measure of complexity. In Foundations of Computer Science, 1977, 18th Annual Symposium on 1977 (pp. 222-227). IEEE.


\end{thebibliography}
